\documentclass[12pt]{article}

\usepackage{upgreek}

\usepackage{hyperref}

\usepackage[normalem]{ulem}
\usepackage{amsthm}
\usepackage{amsmath}
\usepackage{amssymb}
\usepackage{mathrsfs}
\usepackage{mathtools}
\usepackage{graphicx}
\usepackage{tikz}
\usetikzlibrary{arrows}
\usepackage{enumerate}
\usepackage{comment}

\usepackage{fullpage}

\usepackage[ruled,vlined]{algorithm2e}

\usepackage{appendix}
\usepackage{todonotes}

\newtheorem{theorem}{Theorem}

\newtheorem{lemma}[theorem]{Lemma}

\newtheorem{corollary}[theorem]{Corollary}

\theoremstyle{definition}
 
\theoremstyle{remark}

\theoremstyle{remark}

\usepackage{xcolor}

\numberwithin{theorem}{section}

\providecommand{\N}{}

\renewcommand{\N}{{\mathbb N}}

\newcommand{\E}{\textsf{\upshape E}}
\newcommand{\prob}{\textsf{\upshape Pr}}

\begin{document}

\title{A Direct Proof of the Short-Side Advantage\\[2mm] in Random Matching Markets}

\author{
Simon Mauras\thanks{Inria, FairPlay joint team, Palaiseau, France; e-mail: \texttt{simon.mauras@inria.fr}}
\and
Pawe\l{} Pra\l{}at\thanks{Department of Mathematics, Toronto Metropolitan University, Toronto, ON, Canada; e-mail: \texttt{pralat@torontomu.ca}}
% Pawe\l{} Pra\l{}at\thanks{Decision Analysis and Support Unit, SGH Warsaw School of Economics, Warsaw, Poland; e-mail: \texttt{pprala@sgh.waw.pl}}
\and 
Adrian Vetta\thanks{Department of Mathematics \& Statistics, and School of Computer Science, McGill University, Montreal, QC, Canada; e-mail: \texttt{adrian.vetta@mcgill.ca}}
}

\maketitle

\begin{abstract}
We study the stable matching problem under the random matching model where the preferences of the doctors and hospitals are sampled uniformly and independently at random. In a balanced market with $n$ doctors and $n$ hospitals, the doctor-proposing deferred-acceptance algorithm gives doctors an expected rank of order $\log n$ for their partners and hospitals an expected rank of order $\frac{n}{\log n}$ for their partners~\cite{Pit89,Wil72}. This situation is reversed in an unbalanced market with $n+1$ doctors and $n$ hospitals~\cite{AKL17}, a phenomenon known as the short-side advantage. The current proofs~\cite{AKL17,CT22} of this fact are indirect,  counter-intuitively being based upon analyzing the hospital-proposing deferred-acceptance algorithm. In this paper we provide a direct proof of the short-side advantage, explicitly analyzing the doctor-proposing deferred-acceptance algorithm. Our proof sheds light on how and why the phenomenon arises.
\end{abstract}

\section{Introduction}

In this paper, we study the deferred-acceptance algorithm for the stable matching problem. In the balanced case with $n$ doctors and $n$ hospitals, it is known that the proposing side has a significant advantage in terms of the ranks of their assigned partners. Take, for example, a random matching market, where the preference lists of each doctor and each hospital are complete lists, sampled uniformly and independently at random. Then, under the doctor-proposing deferred-acceptance algorithm, each doctor has an expected rank for its partner of order $\log n$ while each hospital has an expected rank for its partner of order $\frac{n}{\log n}$~\cite{Pit89, Wil72}. Recently, Ashlagi et al.~\cite{AKL17} proved that this advantage is reversed in an unbalanced market with $n+1$ doctors and $n$ hospitals. Now, the doctor-proposing deferred-acceptance algorithm produces a stable matching where each hospital is matched to a doctor it ranks at the order of $\log n$ on average and each doctor, excluding the one unmatched (unemployed) doctor, is matched to a hospital it ranks at the order of $\frac{n}{\log n}$ on average.

The proof of Ashlagi et al.~\cite{AKL17} (and its simplification by Cai and Thomas~\cite{CT22}) are indirect. Essentially, by studying the hospital-proposing deferred-acceptance algorithm, they show the short-side (hospitals) holds this advantage in every stable matching. In a very technical paper, Pittel~\cite{Pit19} strengthen the bounds on the average rank within each side, proving concentration using an integral formula for its probability mass function. In this paper, we present a direct proof of the short-side advantage by directly studying the doctor-proposing deferred-acceptance algorithm, at the cost of weaker constants in the asymptotic bounds. Arguably, this is the most natural way to prove this result. More importantly, it additionally sheds light on why this scenario is drastically different than the balanced marked. Indeed, as stated in~\cite{CT22}, ``\emph{because some hospital must go unmatched, the algorithm will only terminate once some hospital has proposed to every doctor. This is a very different random process, and one can imagine it must run for a much longer time, forcing the proposing hospitals to be matched to much worse partners than in the balanced case. Unfortunately, this random process is fairly difficult to analyze (for instance, to get a useful analysis, one would need to keep track of which hospital is currently proposing, which doctors they have already proposed to, and how likely each doctor is to accept a new proposal)}"; consequently, they take a different approach.

\bigskip

The paper is structured as follows. To begin, we present some background on the stable matching problem and introduce the random matching model in Section~\ref{sec:background}. Our result and a high-level overview of the argument used is given in Section~\ref{sec:result}. Section~\ref{sec:proof} contains the proof.

\section{Background}\label{sec:background}

\subsection{The Stable Matching Problem}

In the {\em stable matching problem}, we are given a set $D=\{d_1,d_2,\dots, d_m\}$ of $m$ doctors and a set $H=\{h_1,h_2,\dots, h_n\}$ of $n$ hospitals. Let $\mu$ be a matching between the doctors and the hospitals. We say doctor $d$ is matched to hospital $\mu(d)$ in the matching $\mu$, where $\mu(d)=\emptyset$ if $d$ is unmatched. Similarly, hospital $h$ is matched to doctor $\mu(h)$, where $\mu(h)=\emptyset$ if $h$ is unmatched.

Every doctor $d \in D$ has a preference ranking $\succ_d$ over the hospitals; similarly, every hospital $h \in H$ has a preference ranking~$\succ_h$ over the doctors. We assume that every doctor and every hospital prefers to be matched than unmatched. Given the preference rankings, we say that doctor~$d$ and hospital~$h$ form a {\em blocking pair} $\{d,h\}$ if they prefer each other to their partners in the matching~$\mu$; that is, $h \succ_d \mu(d)$ and $d \succ_h \mu(h)$. A matching $\mu$ that contains no blocking pair is called {\em stable}; otherwise it is unstable. Consequently, the set of stable matchings form the core of the {\em stable matching game}.\footnote{The {core} of a game is the set of feasible solutions that have no {\em blocking coalition}, a subset of agents that can all benefit by forming their own solution away from the grand coalition. For that reason, core solutions are considered ``stable". In the stable matching game, any blocking coalition must contain a blocking coalition of cardinality two, that is, a {\em blocking pair}.}

\subsection{The Stable Matching Lattice}\label{sec:lattice}

Critically, the core in this game is non-empty~\cite{GS62}. Indeed, the core may contain an exponential number of stable matchings~\cite{Knu82}. Intriguingly, the core induces a {\em distributive lattice} $(\mathcal{M}, \geqslant)$ as noted by Conway (see Knuth~\cite{Knu82}). Formally, the order $\geqslant$ is defined via the preference lists of the doctors by setting $\mu_1 \geqslant \mu_2$ if and only if every doctor weakly prefers its partner in the stable matching $\mu_1$ over its partner in the stable matching $\mu_2$; that is $\mu_1(d) \succeq_d \mu_2(d)$, for every doctor~$d$. Interestingly, McVitie and Wilson~\cite{MW71} proved that the interests of the doctors and the hospitals are diametrically opposed within the lattice. That is, every doctor weakly prefers its partner in the stable matching $\mu_1$ to its partner in the stable matching $\mu_2$ if and only if every hospital weakly prefers its partner in the stable matching $\mu_2$ to its partner in the stable matching $\mu_1$. By the lattice property, in the {\em supremum} of the lattice each doctor is matched to its most preferred partner amongst any stable matching (called its {\em best stable-partner}) and each hospital is matched to its least preferred partner amongst any stable matching (called its {\em worst stable-partner}). Thereupon, the supremum stable matching is called the {\em doctor-optimal/hospital-pessimal} stable matching. Conversely, the {\em infimum} of the lattice is the {\em doctor-pessimal/hospital-optimal} stable matching.

\subsection{The Deferred-Acceptance Algorithm}\label{sec:dfa}

The celebrated {\em deferred-acceptance algorithm} by Gale and Shapley~\cite{GS62} outputs the doctor-optimal stable matching when the doctors make proposals (see Algorithm~\ref{alg:DFA}) and the hospital-optimal stable matching when the hospitals propose. These results hold regardless of the specific order of proposals.

\begin{algorithm}
\While{there is an unmatched doctor $d$ who has not proposed to every hospital}{
Let $d$ propose to its favourite hospital $h$ who has not yet rejected it\;
  \uIf{$h$ is unmatched}{
    $h$ provisionally matches with $d$\;
  }
  \uElseIf{$h$ is provisionally matched to $\hat{d}$}{
    $h$ provisionally matches to its favourite of $d$ and $\hat{d}$, and rejects the other\;
  }
  }
\caption{Doctor-Proposing Deferred-Acceptance Algorithm}\label{alg:DFA}
\end{algorithm}

Of subsequent relevance is that the deferred-acceptance algorithm terminates either when every doctor is matched or when every unmatched doctor has proposed to every hospital. For a comprehensive study of stable matchings see, for example, the book of Gusfield and Irving~\cite{GI89}.

\subsection{The Random Matching Model}\label{sec:random-model}

The deferred-acceptance algorithm has been widely studied in {\em random matching model} beginning with Wilson~\cite{Wil72} and further classical works by Knuth, Pittel and coauthors~\cite{Knu82,Pit89,KMP90,Pit92}. In the random matching model, the ranking of each doctor is a random permutation of the hospitals, namely the set $[n]=\{1,2,\dots,n\}$, drawn uniformly and independently at random. Similarly, the ranking of each hospital is a random permutation of the set $[m]=\{1,2,\dots,m\}$, drawn uniformly and independently at random.

Consider first the case of a {\em balanced market} where the number of doctors equals the number of hospitals, that is, when $m=n$. Under the doctor-proposing deferred acceptance algorithm, with high probability, each doctor ranks its partner at position of order $\log n$ on average in their preference ranking. In sharp contrast, on average, each hospital only ranks its partner at position of order $\frac{n}{\log n}$~\cite{Wil72,Pit89}.

However, in an {\em unbalanced market} the situation is reversed. Suppose that there is one more doctor than hospital, that is, $m=n+1$. In a remarkable result, Ashlagi et al.~\cite{AKL17} showed that, with high probability, each hospital ranks its partner at $O(\log n)$ on average and each doctor ranks its partner at $\Omega(\frac{n}{\log n})$ on average, despite the use of the doctor-proposing deferred-acceptance algorithm. Indeed, every stable matching has this property. Thus, in an unbalanced market, the proposing side loses its advantage if its side is {\em long} (that is, of greater cardinality than the proposal-receiving side) and the advantage transfers to the {\em short} side. 

Finally, the study of unbalanced random matching markets is closely related to the \emph{core-convergence} phenomenon: empirically, the set of stable matching is small, which mitigates the manipulability of the deferred acceptance mechanism. This phenomenon was first observed by Roth and Peranson~\cite{roth1999redesign}, and can be explained in markets with unbalance~\cite{AKL17,Pit19}, with short preference lists~\cite{immorlica2015incentives,kojima2009incentives}, and with correlated preferences~\cite{coles2014optimal,holzman2014matching,lee2016incentive,azevedo2016supply,gimbert2021two}, but is mitigated in markets exhibiting local structures~\cite{biro2022large,rheingans2024large}.

\section{Our Result}\label{sec:result}

Analyzing the doctor-proposing deferred-acceptance algorithm directly is much simpler in the balanced case than in the unbalanced case. This is because in balanced markets the algorithm terminates when each hospital has received at least one proposal. Thus, the problem corresponds to a coupon collecting problem, where a coupon is collected each time a hospital receives a proposal. More precisely, allowing for redundant proposals, the sequence of proposals made is uniform, and its length upper bounds the total rank of hospitals from the perspective of doctors; see \cite{Wil72,CT22}. In contrast, in an unbalanced market (with $n+1$ doctors) the algorithm terminates when some doctor has been rejected by every hospital. Ergo, the number of proposals made is far greater in an unbalanced market than in a balanced market. In particular, adding one doctor to a balanced market creates a long sequence of proposals, switching to a new doctor each time a proposal is accepted, and stopping only when the proposing doctor is rejected by all hospitals. As a result, this innocent looking change triggers a ``domino effect'' that affects not only the rejected doctor but also others. Consequently, understanding the sequence of proposals made under the doctor-proposing deferred-acceptance algorithm in an unbalanced random matching market is complex. Indeed, Cai and Thomas~\cite{CT22} state that ``this random process is fairly difficult to analyze" because it would entail keeping ``track of which doctor is proposing, which hospital they have already proposed to, and how likely each hospital is to accept a new proposal".

For this reason, Ashlagi et al.~\cite{AKL17} prove their result indirectly. Specifically, rather than doctor-proposing deferred-acceptance algorithm they study the hospital-proposing deferred-acceptance algorithm. Informally, their argument is as follows. Because the hospitals form the short side this terminates quickly with the hospital-optimal stable matching. They then show that even if a doctor rejects every proposal it receives it is extremely unlikely to receive many proposals. Thus the expected rank of its {\em best stable partner} is not good. This, in turn, implies that in \emph{every} stable matching the doctor is matched with a partner of low preference. Of course, the doctor-optimal stable matching is one of these stable matchings. So this conclusion must apply if we had instead run the doctor-proposing deferred-acceptance algorithm to find it. The proof in~\cite{AKL17} is quite long and intricate. Cai and Thomas~\cite{CT22} provide a nice, short proof but, again, their proof is indirect based on the outline above for hospital-proposings. Recent works by Kanoria et. al.~\cite{kanoria2021matching}, and Potukuchi and Singh~\cite{PS24} extend this approach to the setting in which agents have incomplete preference lists.

In this paper we present a direct proof that specifically analyzes the doctor-proposing deferred-acceptance algorithm. Let us formally state the result and then provide a high level overview of the proof strategy. 

\subsection{Statement of the Result} 

There are $n$ \emph{hospitals} that form a set $H = \{ h_1, h_2, \ldots, h_n\}$ and $(n+1)$ \emph{doctors} that form a set $D = \{ d_1, d_2, \ldots, d_{n+1}\}$. Each hospital $h \in H$ has a random preference over doctors that can be expressed by a permutation $r_h : D \to [n+1]$ taken uniformly at random from the set of all $(n+1)!$ permutations of the set $[n+1]=\{1, 2, \ldots, n+1\}$; $r_h(d)$ is the rank of a doctor $d$ for a hospital $h$ (the smaller the better). Similarly, each doctor $d \in D$ has a random preference over hospitals that can be expressed by a permutation $r_d : H \to [n]$ taken uniformly at random from the set of all $n!$ permutations of the set $[n]$. All $(2n+1)$ preferences are generated independently. 

Our results are asymptotic by nature, that is, we will assume that $n \to \infty$. We say that some event holds \emph{asymptotically almost surely} (\emph{a.a.s.}) if it holds with probability tending to one as $n \to \infty$. Our goal is to understand a typical ``happiness'' of the players once stable matching is established. To that end, we define the average ranks of doctors (from the perspective of hospitals they are matched to) and the average ranks of hospitals (from the perspective of doctors), that is,
$$
A_H(\mu) = \frac {1}{n} \sum_{h \in H} r_h ( \mu(h) ) \qquad \text{and} \qquad A_D(\mu) = \frac {1}{n+1} \sum_{d \in D} r_d ( \mu(d) ),
$$
where $\mu$ is a matching produced by the doctor-proposing deferred-acceptance algorithm.

Now, we are ready to state the main result.

\begin{theorem}
Consider a random matching market, where the preference lists of each of the $n+1$ doctors and each of the $n$ hospitals are sampled uniformly and independently at random. Then, the doctor-proposing deferred-acceptance algorithm produces a stable matching $\mu$ where each hospital is matched to a doctor it ranks at the order of $\log n$ on average and each doctor is matched to a hospital it ranks at the order of $\frac{n}{\log n}$ on average, a.a.s. In other words, a.a.s.
$$
A_H(\mu) = \Theta \left( \log n \right) \qquad \text{and} \qquad A_D(\mu) = \Theta \left( \frac {n}{\log n} \right).
$$
\end{theorem}

\subsection{Sketch of the Argument}

Let 
$$
k = \left\lfloor \frac {n}{5c \log n} \right\rfloor \qquad \text{and} \qquad \ell = \left\lfloor \frac {n^2} {ac \log n} \right\rfloor, 
$$
where $c = 18 > \max\{8, 4 / \log(5/4)\}$ and $a = 100 e^{200} \ge \max\{150, 300, 100 e^{200}\}$ are constants that are large enough to satisfy various inequalities that will be required for the argument to hold. We remark that we did not attempt to optimize the constants. 

We consider an algorithm in which doctors propose. We say that a hospital $h \in H$ is \emph{popular} if it received at least $k$ proposals. Hospitals that are not popular are called \emph{unpopular}. We will stop the algorithm prematurely at time $T$ when there are exactly $\lfloor c \log n / 4 \rfloor$ popular hospitals. Of course, the algorithm might converge to a stable matching before this happens, so we are not guaranteed to reach this situation. However, we will show that it does happen a.a.s.\ (see Corollary~\ref{cor:stop}). 

To that end we argue as follows. We say that a hospital $h \in H$ \emph{likes} a doctor $d \in D$ if $d$ is in one of the top $\lfloor c \log n \rfloor$ places on its list of preferences. Note that whether $h$ likes $d$ or not depends only on its list of preferences $r_h$; in particular, it is not affected by whether $d$ proposed to $h$ or not. We will first show that a.a.s.\ each doctor $d \in D$ is liked by at least $c \log n / 2$ hospitals (see Lemma~\ref{lem:liking}). Since at time $T$ there are at most $\lfloor c \log n / 4 \rfloor$ popular hospitals in total (exactly that many if the process stopped prematurely), there are thus at least $c \log n / 4$ unpopular hospitals that like $d$. It is unlikely that they are all matched with someone better than $d$. Indeed, we will show that a.a.s.\ this does not happen for {\em any} of the doctors (see Lemma~\ref{lem:no_match}). But, if a stable matching is reached, then the unique doctor that is unemployed at that time would have made a proposal to every hospital. The fact that they are still unemployed means that all of the hospitals are matched with someone better than~$d$. Since it is not true a.a.s., it implies that a.a.s.\ the process stopped  prematurely at time $T$ when the number of popular hospitals reached $\lfloor c \log n / 4 \rfloor$, not because a stable matching was found (again, see Corollary~\ref{cor:stop}). 

The argument so far shows that a.a.s.\ at least $\lfloor c \log n / 4 \rfloor$ hospitals become popular before the algorithm converges to a stable matching. This means that, before the algorithm finishes, a small set of $\lfloor c \log n / 4 \rfloor$ hospitals receives at least $k \cdot \lfloor c \log n / 4 \rfloor = (1+o(1))\cdot n/20$ proposals in total. Our next task is to show that such a situation is a.a.s.\ impossible if the total number of proposals is exactly $\ell$, regardless of which doctors make such proposals (see Lemma~\ref{lem:lower_bound}). (Clearly, this also implies that it cannot be done with less than $\ell$ proposals a.a.s.) Therefore a.a.s.\ doctors propose (on average) at least $\ell/(n+1) = \Omega(n/\log n)$ times before a stable matching is found. 

Finally, we show that a.a.s.\ less than $\lfloor 400 a c \log n \rfloor$ hospitals receive at most $n / (20 ac \log n)$ proposals if the total number of proposals is exactly $\ell$, regardless of which doctors make such proposals (see Lemma~\ref{lem:upper_bound}). This implies that a.a.s.\ hospitals are matched with doctors that are in the top $O(\log n)$ places of their preference lists (on average).

\section{The Direct Proof}\label{sec:proof}

\subsection{Concentration Inequalities}
Let us first state a few specific formulations of Chernoff's bound that will be useful. Let $X \in \textrm{Bin}(n,p)$ be a random variable distributed according to the Binomial distribution with parameters $n$ and $p$. Then, a consequence of \emph{Chernoff's bound} (see e.g.~\cite[Theorem~2.1]{JLR}) is that for any $t \ge 0$ we have
\begin{eqnarray}
	\prob( X \ge \E (X) + t ) &\le& \exp \left( - \E (X) \cdot \varphi \left( \frac {t}{\E (X)} \right) \right) ~~\le~~ \exp \left( - \frac {t^2}{2 (\E (X) + t/3)} \right)  \label{chern1} \\
	\prob( X \le \E (X) - t ) &\le& \exp \left( - \E (X) \cdot \varphi \left( \frac {-t}{\E (X)} \right) \right) ~~\le~~ \exp \left( - \frac {t^2}{2 \E (X)} \right),\label{chern}
\end{eqnarray}
where $\varphi(x) = (1+x) \log (1+x) - x$, for $x > -1$ (and $\varphi(x) = \infty$ for $x \le -1$). 

Moreover, these bounds hold in a more general setting as well, that is, for $X=\sum_{i=1}^n X_i$ where $(X_i)_{1\le i\le n}$ are independent variables and for every $i \in [n]$ we have $X_i \in \textrm{Bernoulli}(p_i)$ with (possibly) different $p_i$-s (again, see~e.g.~\cite{JLR} for more details).

\subsection{Observation 1}

We will use the well-known principle of deferred decision, introduced in~\cite{KMP90}, where the random preferences are drawn when the algorithm needs them and not a priori, simplifying the analysis. More precisely, we may defer exposing information about the random permutation $r_h$ for a given hospital $h$ until a new proposal is made to $h$. When $d$ proposes to $h$, we simply place $d$ in a random place on $h$'s list of preferences that is still available. 

Moreover, note that permutations of $[n+1]$ can be generated in the following unorthodox, but convenient, way. We first select a random subset $A$ of $[n+1]$ of cardinality $k$, permute independently elements of $A$ and $[n+1] \setminus A$, and then concatenate these two permutations (of lengths $k$ and $(n+1-k)$, respectively) to create a permutation of $[n+1]$. Translating this to our scenario, for each hospital $h \in H$, we may independently generate a random subset $A_h$ of $[n+1]$ of cardinality $k$, before the matching algorithm starts. Each time a new proposal is made to $h$, the doctor that proposed is placed at a random place from $A_h$ that is still available in $h$'s permutation. Once we run out of available spots in $A_h$ (which means that $h$ just became popular), we move to selecting spots from $[n+1] \setminus A_h$. However, for unpopular hospitals we only ever draw spots from $A_h$ giving us the following observation. 

\begin{lemma}\label{lem:deferred}
For each hospital $h \in H$, we independently generate a subset $A_h \subseteq [n+1]$ of cardinality $k$. We run the algorithm until it stops prematurely or a stable matching is created. Let $D_h \subseteq D$ be the set of doctors that proposed to hospital $h$. 
Then, the following property holds: for each unpopular hospital $h \in H$, doctors in $D_h$ have ranks from $A_h$ on the list of preferences of hospital $h$. \qed
\end{lemma}

\subsection{Observation 2}

If the game ends because a stable matching is created, then some poor doctor $d$ must have proposed to every single hospital but is still unemployed. Most hospitals have this doctor $d$ not so high on their lists so it is not too surprising that they did not like them. However, there are still many hospitals that have $d$ quite high on their respective preference lists. It is unlikely that all of them, especially unpopular ones, found a better match. This motivates the following definition. 

Recall that a hospital $h \in H$ \emph{likes} a doctor $d \in D$ if $d$ is in one of the top $\lfloor c \log n \rfloor$ places on its list of preferences. For a given doctor $d \in D$, the number of hospitals that like $d$ is the binomial random variable $X \in \textrm{Bin}(n, \lfloor c\log n \rfloor / (n+1))$ with expectation asymptotic to $c \log n$. It follows immediately from Chernoff's bound~(\ref{chern}), applied with $t = \E (X) -  c \log n / 2 \sim c \log n / 2$, that there are at most $c \log n / 2$ hospitals that like $d$ with probability upper bounded by 
$$
\exp \left( - (1+o(1)) \frac {c}{8} \log n \right) = o(1/n), 
$$
provided that $c > 8$. Hence, by the union bound over $(n+1)$ doctors, we get the following.

\begin{lemma}\label{lem:liking}
Every doctor is liked by at least $c \log n / 2$ hospitals, a.a.s. \qed
\end{lemma}

\subsection{Observation 3}

Let us concentrate on a fixed doctor $d \in D$. We first go over hospitals' preference lists, regardless of whether $d$ proposed to them or not, to identify hospitals that like it. By Lemma~\ref{lem:liking}, since we aim for a statement that holds a.a.s., we may assume that there are at least $c \log n / 2$ hospitals that like $d$. We may arbitrarily select $\lceil c \log n / 2 \rceil$ of them.

Recall that we stop the process prematurely when there are exactly $\lfloor c \log n / 4 \rfloor$ popular hospitals, in total. Hence, at time $T$ there are at least $\lfloor c \log n / 4 \rfloor$ unpopular hospitals that like~$d$. Unfortunately, which of the hospitals that like $d$ are unpopular at time $T$ depends on many other events. Consequently, to ensure we deal with all possibilities, we will take a union bound over all possible selections of $\lfloor c \log n / 4 \rfloor$ hospitals from among our prior selection of $\lceil c \log n / 2 \rceil$ hospitals. 

By Lemma~\ref{lem:deferred}, doctors that proposed to a unpopular hospital $h$  (doctors in the set $D_h$) have ranks on the list of preferences of $h$ from a random set $A_h \subseteq [n+1]$ of cardinality $k$. We will investigate these sets $A_h$ to estimate the probability that the corresponding hospitals are matched to someone better than~$d$. The probability that a given unpopular hospital $h$ that likes $d$ is matched with someone better than $d$ is no greater than the probability that $h$ likes someone from those doctors that proposed to $h$, other than $d$, which in turn is no greater than the probability that $h$ likes someone from $D_h \setminus \{d\}$ (let us use $E_h$ for this event).  The probability that $h$ does \emph{not} like anyone from $D_h \setminus \{d\}$ is lower bounded by
\begin{eqnarray*}
\prod_{i=1}^{k} \left( 1 - \frac {\lfloor c\log n \rfloor }{n-i} \right) &=& \left( 1 - (1+o(1)) \frac {c\log n}{n} \right)^{k}\\ 
&=& \exp\left( - (1+o(1)) \frac {c\log n}{n} \cdot k \right) \\
&=& (1+o(1)) e^{-1/5}.
\end{eqnarray*}
(Recall that $(1-x)=\exp(x+O(x^2))$ by a Taylor expansion.)
Hence, the probability that we aimed to estimate is less than $1 - (1+o(1)) e^{-1/5} < 1/5$.
Since the events $E_h$ corresponding to the selected unpopular hospitals are independent, the probability that all unpopular hospitals that like $d$ are matched with someone better than $d$ is upper bounded by 
\begin{eqnarray*}
\binom{\lceil c \log n / 2 \rceil} {\lfloor c \log n / 4 \rfloor} \left( \frac {1}{5} \right)^{\lfloor c \log n / 4 \rfloor} 
&\le& 2^{\lceil c \log n / 2 \rceil} \left( \frac {1}{5} \right)^{\lfloor c \log n / 4 \rfloor} \\
&=& O(1) \cdot \left( \frac {4}{5} \right)^{c \log n / 4} \\
&=& O(1) \cdot \exp \left( - \frac {c \log(5/4)}{4} \log n \right)\\
&=& o(1/n), 
\end{eqnarray*}
provided that $c > 4 / \log(5/4) \approx 17.92$. (Note that, trivially,  for any two integers $1 \le a \le b$, $\binom{a}{b} \le 2^a$.)
Hence, by the union bound over all doctors, we get the following. 

\begin{lemma}\label{lem:no_match}
The following property holds at time $T$ a.a.s.: for every doctor $d$ there exists at least one unpopular hospital that is not matched with a doctor it prefers over $d$. \qed
\end{lemma}

The above lemma implies immediately that the process had to stop prematurely at time~$T$, before a stable matching is found. Hence, we obtain:

\begin{corollary}\label{cor:stop}
The algorithm stops prematurely at time $T$ when there are exactly $\lfloor c \log n / 4 \rfloor$ popular hospitals, a.a.s. \qed
\end{corollary}

\subsection{Observation 4}

Corollary~\ref{cor:stop} implies that a.a.s.\ at least $\lfloor c \log n / 4 \rfloor$ hospitals become popular before the algorithm converges to a stable matching. By definition, each of these popular hospitals received at least $k = \lfloor n/(5c \log n) \rfloor$ proposals. The next lemma shows that it takes at least $\ell = \lfloor n^2 / (ac \log n) \rfloor$ proposals, in total, to reach this situation.

\begin{lemma}\label{lem:lower_bound}
The following property holds a.a.s.: no set of $\lfloor c \log n / 4 \rfloor$ hospitals received at least $k \cdot \lfloor c \log n / 4 \rfloor = (1+o(1)) \, n/20$ proposals, given that the total number of proposals is $\ell = \lfloor n^2 / (ac \log n) \rfloor$ (regardless of how many proposals each doctor made).
\end{lemma}
\begin{proof}
Suppose that each doctor $d \in D$ proposed $x_d$ times for a total of $\sum_{d \in D} x_d = \ell = \lfloor n^2 / (ac \log n) \rfloor$ proposals. Of course, $d$ is proposing according to its list of preferences $r_d$ so $x_d$ determines which hospitals $d$ proposes to. Recall the number of positive, integer solutions to $y_1+y_2+\cdots y_K= N$, called the number of {\em compositions} of $N$ into exactly $K$ parts, is equal to $\binom{N-1}{K-1}$. Therefore, setting $K=|D|=n+1$ and $N=\ell$, the number of different scenarios we must consider is 
$$
\binom{\ell -1}{|D|-1} 
= \binom{\ell -1}{n} 
\le \left( \frac {e \ell (1+o(1))}{n} \right)^n 
= \left( \frac {e n (1+o(1))}{ac \log n} \right)^n 
\le n^n 
= \exp( n \log n ),
$$
where the first inequality follows from the fact that for any two integers $1 \le a \le b$, $\binom{a}{b} \le (ea/b)^b$.
Unfortunately, for our purposes, this number of scenarios is too large to successfully apply the union bound over. 

So we must reduce the number of scenarios to consider.
To do this, we will apply the following trick. Consider any scenario $(x_d)_{d \in D}$ with $\sum_{d \in D} x_d = \ell$. We construct an \emph{auxiliary} scenario (based on $(x_d)_{d \in D}$) as follows. If $x_d \le \lceil n / (ac\log n) \rceil$, then we fix $\hat{x}_d = \lceil n / (ac\log n) \rceil$. Otherwise (that is, if $x_d > \lceil n / (ac\log n) \rceil$), we ``round'' $x_d$ up, that is, we fix $\hat{x}_d \ge x_d$ to be the smallest value of the form $2^i \lceil n / (ac\log n) \rceil$ for some $i \in \N$. Trivially, if the original scenario $(x_d)_{d \in D}$ makes at least $k \cdot \lfloor c \log n / 4 \rfloor$ proposals to some set of $\lfloor c \log n / 4 \rfloor$ hospitals, then the auxiliary scenario $(\hat{x}_d)_{d \in D}$ does so too. More importantly, the total number of proposals of the auxiliary scenario is not much larger than the original value, namely, $\ell$. Indeed, note that
$$
\sum_{d \in D} \hat{x}_d \le (n+1) \cdot \lceil n / (ac\log n) \rceil + 2 \sum_{d \in D} x_d \le (1+o(1)) \, 3 \ell.
$$
The advantage is that there are substantially fewer auxiliary scenarios than in the original case. 

For any $i \in \N$, let $z_i$ be the number of values of $\hat{x}_d$ that are at least $2^i \lceil n / (ac\log n) \rceil$. Since at most $n/2^{i-1}$ values of $x_d$ are more than $2^{i-1} \lceil n / (ac\log n) \rceil$, we have $z_i \le n/2^{i-1}$. Hence, the number of auxiliary scenarios to consider can be estimated as follows:
\begin{eqnarray*}
\sum_{z_1 \ge z_2 \ge \ldots} \binom{n+1}{z_1} \binom{z_1}{z_2} \binom{z_2}{z_3} \cdots &\le& \sum_{z_1 \ge z_2 \ge \ldots} 2^{n+1} \cdot 2^{z_1} \cdot 2^{z_2} \cdots \\
&\le& \sum_{z_1 \ge z_2 \ge \ldots} 2^{(n+1)+n+n/2+n/4+\ldots} \\
&\le& (n+1)^{O(\log n)} \cdot 2^{3n} \\
&=& \exp( O( \log^2 n) ) \cdot 2^{3n}.
\end{eqnarray*}
The justification for the first and the second inequality is as follows: since $z_i \le n/2^{i-1}$, $\binom{z_i}{z_{i+1}} \le 2^{z_i} \le 2^{n/2^{i-1}}$. 
The third inequality holds because there are at most $\log_2 n$ terms $z_i$ and each of them is no greater than $n+1$. 

Finally, the number of choices for the set of $\lfloor c \log n / 4 \rfloor$ hospitals is 
$$
\binom{n}{\lfloor c \log n / 4 \rfloor} \le \left( \frac {en}{\lfloor c \log n / 4 \rfloor} \right)^{\lfloor c \log n / 4 \rfloor} = \exp( O( \log^2 n) ).
$$
Hence, the union bound has to be taken over $\exp( O( \log^2 n) ) \, 2^{3n}$ pairs consisting of some set of $\lfloor c \log n / 4 \rfloor$ hospitals and some auxiliary configuration. 

Let us then fix any set of $\lfloor c \log n / 4 \rfloor$ hospitals and any auxiliary configuration $(\hat{x}_d)_{d \in D}$ with $\sum_{d \in D} \hat{x}_d \le (1+o(1)) \, 3 \ell$. Before we estimate how many of the $\sum_{d \in D} \hat{x}_d$ proposals are made to the selected hospitals, we need to deal with one more technicality, namely, \emph{active} doctors that proposed at least $n/2$ times. Active doctors in the extreme case may propose to all the hospitals including the selected ones. But, clearly, there cannot be too many active doctors. Indeed, let $i \in \N$ be the largest $i$ such that $2^i \le ac \log n / 3$ (which implies that $2^i \ge ac \log n / 6$). We get that the number of values of $\hat{x}_d$ that are at least $2^i \lceil n / (ac\log n) \rceil \le (1+o(1)) \, n/3$ is $z_i \le n/2^{i-1} \le 12 n / (ac \log n)$. Hence, the number of active doctors is at most $12 n / (ac \log n)$ and so they contribute at most $12 n / (ac \log n) \cdot \lfloor c \log n / 4 \rfloor \le 3n/a \le n / 50$ to the total number of $(1+o(1)) \, n/20$ proposals made to the selected hospitals, provided that $a \ge 150$. It follows that at least $n/40$ proposals made to the selected hospitals have to come from non-active doctors. 

We expose now proposals from non-active doctors, one by one. For a given non-active doctor $d \in D$, we expose hospitals in the top $\hat{x}_d$ positions on the preference list of $d$, one proposal at the time. The probability that a proposal is made to a hospital from the selected ones is equal to
$$
\frac {\lfloor c \log n / 4 \rfloor - i}{n-j} \le \frac {\lfloor c \log n / 4 \rfloor}{n/2} \le \frac {c \log n}{2n},
$$
where $i$ is the number of the selected hospitals $d$ already proposed to and $j$ is the number of hospitals $d$ proposed to so far; clearly, $j \le n/2$, since $d$ is not active. As mentioned earlier, the total number of proposals (coming from non-active doctors) is at most $(1+o(1)) \, 3\ell \le 4 \ell$. Hence, the number of proposals that are made to the selected hospitals can be stochastically upper bounded by the binomial random variable $X \in \text{Bin}(4 \ell, c \log n / (2n))$ with $\E (X) = 2 \ell c \log n / n = 2n / a - o(1)$ and $\E (X) \le 2n / a$. We apply the (slightly more fancy) Chernoff's bound~(\ref{chern1}) with $t = n/40 - \E (X) \ge n/40 - n/150 > n/50$, provided that $a \ge 300$. Note that
\begin{eqnarray*}
\varphi (t / \E (X)) &\ge& \varphi(a/100)\\ 
&=& (1+a/100) \log (1+a/100) - a/100 \\
&\ge& (a/100) \log (a/100) - a/100 \\
&\ge& (a/100) \log (a/100) \left(1 - \frac {1}{\log(a/100)} \right) \\
&\ge& 3a/2,
\end{eqnarray*}
provided that $a \ge 100 e^{200}$. 
We get that the failure probability is upper bounded by 
$$
\exp \left( - \E (X) \cdot \varphi(t / \E (X)) \right) \le \exp \left( - (1+o(1)) \frac {2n}{a} \cdot \frac {3a}{2} \right) = \exp( - (1+o(1)) \, 3n)
$$
and, as promised, the conclusion holds by the union bound. Indeed, the probability that there exists a set of hospitals and an auxiliary configuration that fails the desired property is less than
$$
\exp( O( \log^2 n) ) \, 2^{3n} \exp( - (1+o(1)) \, 3n) = \exp( O( \log^2 n) ) \, \left( \frac {2}{e} + o(1) \right)^{3n} = o(1). \qedhere
$$
\end{proof}

\subsection{Observation 5}

The next lemma implies that less than $\lfloor 400 a c \log n \rfloor$ hospitals received at most $n / (20 ac \log n)$ proposals. These hospitals might be matched with a doctor that is not necessarily at the top of their corresponding lists but, trivially, with the rank of $O(n)$. The remaining hospitals received at least $n / (20 ac \log n)$ proposals and so the rank of their partners is expected to be $O(\log n)$. Hence, the expected rank of a typical hospital is
$$
\frac {O(\log n)}{n} \cdot O(n) + \frac{n-O(\log n)}{n} \cdot O(\log n) = O(\log n).
$$

\begin{lemma}\label{lem:upper_bound}
The following property holds a.a.s.: no set of $\lfloor 400 a c \log n \rfloor$ hospitals received at most $20 n$ proposals, given that the total number of proposals is $\ell = \lfloor n^2 / (ac \log n) \rfloor$ (regardless of how many proposals each doctor made).
\end{lemma}
\begin{proof}
The argument is the same as the one used to prove Lemma~\ref{lem:lower_bound}, except for a few modifications.
Specifically, as before, suppose that a doctor $d \in D$ proposed $x_d$ times for a total of $\sum_{d \in D} x_d = \ell = \lfloor n^2 / (ac \log n) \rfloor$ proposals. Again, to reduce the number of scenarios to consider, we construct an \emph{auxiliary} scenario (based on $(x_d)_{d \in D}$) but this time we ``round'' $x_d$ down instead. If $x_d < \lceil n / (4ac\log n) \rceil$, then we fix $\hat{x}_d = 0$. Otherwise (that is, if $x_d \ge \lceil n / (4ac\log n) \rceil$), we fix $\hat{x}_d \le x_d$ to be the largest value of the form $2^i \lceil n / (4ac\log n) \rceil$ for some $i \in \N \cup \{0\}$. If the original scenario $(x_d)_{d \in D}$ makes at most $20 n$ proposals to some set of $\lfloor 400 a c \log n \rfloor$ hospitals, then the auxiliary scenario $(\hat{x}_d)_{d \in D}$ does so too. The total number of proposals of the auxiliary scenario is not much smaller than $\ell$, the original value. Specifically,
$$
\sum_{d \in D} \hat{x}_d \ge \sum_{d \in D} x_d - (n+1) \cdot \left\lceil \frac {n}{4ac\log n} \right\rceil - \frac 12 \sum_{d \in D} x_d \ge (1+o(1)) \, \ell/4.
$$
As before, the advantage is that there are substantially less auxiliary scenarios than the original ones. For any $i \in \N \cup \{0\}$, let $z_i$ be the number of values of $\hat{x}_d$ that are at least $2^i \lceil n / (4ac\log n) \rceil$. Since at most $4n/2^{i}$ values of $x_d$ are at least $2^{i} \lceil n / (4ac\log n) \rceil$, we have $z_i \le 4n/2^{i}$. Hence, the number of auxiliary scenarios to consider can be estimated as follows:
\begin{eqnarray*}
\sum_{z_0 \ge z_1 \ge \ldots} \binom{n+1}{z_0} \binom{n+1}{z_1} \binom{n+1}{z_2} \binom{z_2}{z_3} \binom{z_3}{z_4}  \cdots &\le& \exp( O( \log^2 n) ) \, 2^{5n}.
\end{eqnarray*}
Since the number of choices for the set of $O( \log n )$ hospitals is $\exp( O( \log^2 n) )$, the union bound has to be taken over $\exp( O( \log^2 n) ) \, 2^{5n}$ pairs consisting of some set of $\lfloor 400 a c \log n \rfloor$ hospitals and some auxiliary configuration. 

Let us fix any set of $\lfloor 400 a c \log n \rfloor$ hospitals and any auxiliary configuration $(\hat{x}_d)_{d \in D}$ with $\sum_{d \in D} \hat{x}_d \ge (1+o(1)) \, \ell/4$. We expose proposals from doctors, one by one. For a given doctor $d \in D$ with $\hat{x}_d > 0$ (which implies that, in fact, $\hat{x}_d \ge \lceil n / (4ac\log n) \rceil$), hospitals in the top $\hat{x}_d$ positions of the preference list of $d$ form a random set of size $\hat{x}_d$ from $H$, the set of hospitals. We order (arbitrarily) the selected hospitals and check if $d$ proposed to them, one hospital at a time following the fixed order. Suppose that we already investigated $i$ hospitals and $d$ made an offer to $j \le i$ of them. The remaining $x_d - j$ hospitals that $d$ made an offer to form a random set from the $n-i$ hospitals we have not checked yet. So, the probability that a proposal is made to the next hospital from the selected ones is equal to
$$
\frac {\binom{n-i-1}{\hat{x}_d-j-1}}{\binom{n-i}{\hat{x}_d-j}} = \frac {\hat{x}_d-j}{n-i} = \frac {\hat{x}_d-O(\log n)}{n-O(\log n)} = (1+o(1)) \frac {\hat{x}_d}{n} \ge \frac {\hat{x}_d}{2n}\,.
$$
Hence, the number of proposals that are made to the selected set of hospitals can be stochastically lower bounded by  the random variable $X=\sum_{d \in D} \sum_{i=1}^{\lfloor 400 a c \log n \rfloor} X_{d, i}$, where $(X_{d,i})$ are independent variables and for every $1 \le i \le \lfloor 400 a c \log n \rfloor$ and $d \in D$ we have $X_{d,i} \in \textrm{Bernoulli}(p_d)$ with $p_d = \hat{x}_d / (2n)$. Clearly,
\begin{eqnarray*}
\E (X) &=& \E \left( \sum_{d \in D} \sum_{i=1}^{\lfloor 400 a c \log n \rfloor} X_{d, i} \right) ~=~ \sum_{d \in D} \, \lfloor 400 a c \log n \rfloor \cdot p_d ~=~ \frac { \lfloor 400 a c \log n \rfloor } { 2n} \sum_{d \in D} \hat{x}_d \\
&\ge& (1+o(1)) \frac {50 ac \ell \log n}{n} ~=~ (1+o(1)) 50n ~\ge~ 40n.
\end{eqnarray*}
It follows from Chernoff's bound~(\ref{chern}) with $t = \E (X) - 20n \ge \E (X) / 2$ (see also the comment right after~(\ref{chern})) that
$$
\prob(X \le 20n) = \prob( X \le \E (X) - t ) \le \exp \left( - \frac {t^2}{2\E (X)} \right) \le \exp \left( - \frac {\E (X)}{8} \right) \le e^{-5n}
$$
and, as before, the conclusion holds by the union bound. 
\end{proof}

\ 

\noindent {\bf Acknowledgement.}
This material is based upon work supported by the National Science Foundation under Grant No. DMS-1928930, while the authors were in residence at the Simons Laufer Mathematical Sciences Institute in Berkeley, California, during the Fall semester of 2023.

\bibliography{refs} 

\end{document}